\documentclass[11pt]{article}
\usepackage{float,amsmath,amsfonts,amssymb,amsthm,color,graphicx}
\usepackage[margin=1in]{geometry}

 \usepackage{calc}
 \usepackage{tikz}
 \usepackage{pgfplots}
\usepackage{cases}

\usetikzlibrary{arrows}
\usetikzlibrary{shapes}
 \usetikzlibrary{calc}
\usetikzlibrary{decorations.pathreplacing}
\usetikzlibrary{arrows,positioning} 
\usetikzlibrary{pgfplots.groupplots}

\newtheorem{thm}{Theorem}

\newtheorem{prop}{Proposition}

\newtheorem{corollary}{Corollary}
\newtheorem{rem}{Remark}
\theoremstyle{definition}

\title{Contraction-based Observers using non-Euclidean Norms\\
with an Application to Traffic Networks}
\author{Samuel Coogan\thanks{S. Coogan (\texttt{sam.coogan@gatech.edu}) is with the School of Electrical and Computer Engineering and the School of Civil and Environmental Engineering at the Georgia Institute of Technology, Atlanta, GA.}\ and Murat Arcak\thanks{M. Arcak (\texttt{\texttt{arcak@eecs.berkeley.edu}}) is with the Department of Electrical Engineering and Computer Sciences at the University of California, Berkeley, CA. } }
\renewcommand\footnotemark{}

\date{}

\begin{document}
\maketitle
\begin{abstract}
  In this note, we study Luenberger-type full-state observers for nonlinear systems using contraction theory. We show that if the matrix measure of a suitably defined Jacobian matrix constructed from the dynamics of the system-observer interconnection is uniformly negative, then the state estimate converges exponentially to the actual state. This sufficient condition for convergence establishes that the distance between the estimate and state is infinitesimally contracting with respect to some norm on the state-space. In contrast to existing results for contraction-based observer design, we allow for contraction with respect to non-Euclidean norms. Such norms have proven useful in applications. To demonstrate our results, we study the problem of observing vehicular traffic density along a freeway modeled as interconnected, spatially homogenous compartments, and our approach relies on establishing contraction of the system-observer interconnection with respect to the one-norm.
\end{abstract}

\section{Introduction}
A dynamical system is \emph{contractive} if any two solutions exponentially converge to one another \cite{LOHMILLER:1998bf, Sontag:2010fk, Forni:2012qe}. The importance of contraction theory in observer design has been previously noted in \cite{LOHMILLER:1998bf, Aghannan:2003wj, Bonnabel:2010fx, Sanfelice:2012uf, Dani:2015ng}. In \cite{LOHMILLER:1998bf}, it is shown how contraction theory is applied for control synthesis and observer design for several example systems. Contraction-based observer design has been proposed in \cite{Aghannan:2003wj, Bonnabel:2010fx} for a class of Lagrangian systems. The paper \cite{Sanfelice:2012uf} establishes necessary and sufficient conditions for an observer to ensure that the distance between system state and the observer estimate is nonincreasing, and \cite{Dani:2015ng} proposes observer design for stochastic systems. In all cases, the contractive properties are with respect to scaled Euclidean norms so that the resulting distance metric on the state space is Riemannian.  

In contrast, in this paper, we use the matrix measure as the basic tool for establishing contractive properties, which accommodates arbitrary vector norms. Thus, this appears to be the first work that proposes contraction-based observer design for non-Euclidean metrics. As seen in Section \ref{sec:an-example-traffic}, such metrics (such as those based on the one-norm and infinity-norm) are natural for certain applications such as traffic flow networks where the dynamics are modeled as a compartmental flow network.

\section{Problem Formulation}
Consider the dynamical system 
\begin{align}
  \label{eq:1}
\dot{x}&=f(t,x) \\
{y}&=g(t,x)
\end{align}
with state $x\in \mathcal{X} \subseteq \mathbb{R}^n$ for positively invariant and convex $\mathcal{X}$,  output $y\in\mathbb{R}^m$, and $t\in[0,\infty)$. We assume $f(t,x)$ and $g(t,x)$ are differentiable in $x$ and that $f(t,x)$ and $g(t,x)$, as well as the Jacobian matrices $\frac{\partial f}{\partial x}(t,x)$ and $\frac{\partial g}{\partial x}(t,x)$, are continuous in $(t,x)$. When $\mathcal{X}$ is not an open set, we understand this to mean that $f(t,x)$ and $g(t,x)$ can be extended to an open superset of $\mathcal{X}$ for which the above conditions hold.

For $L\in\mathbb{R}^{n\times m}$, consider the Luenberger-type observer given by
\begin{align}
  \label{eq:3}
  \dot{\hat{x}}=f(t,\hat{x})+L(g(t,\hat{x})-g(t,x)).
\end{align}
We assume that always $\hat{x}(t)\in \mathcal{X}$ for all $t\in[0,\infty)$, that is, $\mathcal{X}\times \mathcal{X}$ is positively invariant for the system-observer interconnection. Our objective is to establish that $\lim_{t\to\infty}|x(t)-\hat{x}(t)|= 0$ for any solution $(x(t),\hat{x}(t))$ to the observer-system interconnection defined by \eqref{eq:1}--\eqref{eq:3} for some norm $|\cdot|$ on $\mathbb{R}^n$.

To that end, let $|\cdot|$ be a vector norm on $\mathbb{R}^n$, and let $\Vert\cdot\Vert$ denote its induced matrix norm. For $A\in\mathbb{R}^{n\times n}$, the \emph{matrix measure} of $A$, denoted $\mu(A)$, %
is given by \cite{Vidyasagar:2002ly, Desoer:2008bh}
\begin{align}
  \label{eq:26}
  \mu(A)= \lim_{h\to 0^+}\frac{1}{h}(\Vert I+hA\Vert-1).
\end{align}

The main result of this note, presented in Theorem \ref{prop:problem-formulation} below, establishes a sufficient condition for ensuring convergence of the estimator \eqref{eq:3}. In particular, if the matrix measure of a certain Jacobian matrix induced by the system-observer dynamics is uniformly negative on $\mathcal{X}$, then the dynamics satisfy a contractive property that ensures global exponential stability of the zero estimation error set $\{(x,\hat{x})\in\mathcal{X}\times \mathcal{X}:x=\hat{x}\}$.

\begin{thm}
\label{prop:problem-formulation}
  Let $|\cdot|$ be a norm on $\mathbb{R}^n$ and $\mu(\cdot)$ its associated matrix measure. Consider the system-observer interconnection given by \eqref{eq:1}--\eqref{eq:3}. If
  \begin{align}
    \label{eq:20}
 \mu\left(\frac{\partial f}{\partial x}(t,x)+L\frac{\partial g}{\partial x}(t,x)\right)\leq c   
  \end{align}
for some $c\in\mathbb{R}$ for all $x\in\mathcal{X}$ and $t\geq 0$, then 
  \begin{align}
    \label{eq:19}
    |x(t)-\hat{x}(t)|\leq e^{ct}|x(0)-\hat{x}(0)|
  \end{align}
for any solution $(x(t),\hat{x}(t))$ of the system-observer interconnection. In particular, if the above condition holds for $c<0$, then the zero estimation error set $\{(x,\hat{x})\in\mathcal{X}\times \mathcal{X}:x=\hat{x}\}$ is globally exponentially stable.

\end{thm}

The proof of Theorem \ref{prop:problem-formulation} uses the \emph{Clarke generalized derivative} \cite{Clarke:1990kx} to bound the time derivative of a suitably defined Lyapunov function that is not differentiable everywhere. Basic definitions and facts regarding the Clarke generalized derivative are contained in Appendix~\ref{sec:clarke-gener-deriv}. The main idea of the proof is to show that an auxiliary system induced by the system-observer interconnection is a contractive system; definitions and fundamental results for contractive systems are contained in Appendix~\ref{sec:contractive-systems}. Convergence of the observer follows from the contractive properties of this auxiliary system.

\begin{proof}[Proof of Theorem \ref{prop:problem-formulation}]
Define $V(x,\hat{x})=|x-\hat{x}|$. Let $S\subset \mathbb{R}^n$ be the measure-zero set where $|\cdot|$ is not differentiable, and for $z\not\in S$, define $d(z):=\frac{\partial}{\partial z}|z|$. We claim, for all $x,\hat{x}\in\mathcal{X}$ and for almost all $t$,
  \begin{align}
    \label{eq:8}
\dot{V}(t,x,\hat{x})\leq V^o_t(x,\hat{x};\dot{x},\dot{\hat{x}})\leq cV(x,\hat{x}),
  \end{align}
where
\begin{align}
  \label{eq:21}
V^o_t(x,\hat{x};\dot{x},\dot{\hat{x}})=   \limsup_{(z,\hat{z})\to (x,\hat{x}), z-\hat{z}\not\in S}d(z-\hat{z})(f(t,x)-f(t,\hat{x})+L(g(t,x)-g(t,\hat{x}))),
\end{align}
that is, $V^o_t(x,\hat{x};\dot{x},\dot{\hat{x}})$ is the Clarke derivative of $V$ at $(x,\hat{x})$ in the direction of $(\dot{x},\dot{\hat{x}})$ at time $t$.

Assuming the claim to be true, we then have
\begin{align}
  \label{eq:24}
  \frac{d}{dt}|x(t)-\hat{x}(t)|\leq c |x(t)-\hat{x}(t)|
\end{align}
for almost all $t$ for any solution $(x(t),\hat{x}(t))$ of the system-observer interconnection. Then \eqref{eq:19} follows from the Comparison Lemma \cite [Lemma 3.4, p. 102]{khalil}.

We now prove each inequality of the claim \eqref{eq:8}. When $\dot{V}(t,x,\hat{x})$ exists, it is given by
\begin{align}
  \label{eq:22}
  \dot{V}(t,x,\hat{x})=\frac{\partial V}{\partial x}\dot{x}+\frac{\partial V}{\partial \hat{x}}\dot{\hat{x}}= d(x-\hat{x})(f(t,x)-f(t,\hat{x})-L(g(t,\hat{x})-g(t,x))),
\end{align}
and thus the first inequality of \eqref{eq:8} holds.

For the second inequality, consider the auxiliary system $\dot{\xi}=f(t,\xi)+Lg(t,\xi)=:F(t,\xi)$. By hypothesis, $\mu(F(t,\xi))\leq c$ for all $t\geq 0$ and all $\xi\in\mathcal{X}$. Then, by Corollary \ref{cor:2} in the Appendix,
\begin{align}
  \label{eq:25}
&\limsup_{(z_1,z_2)\to (\xi_1, \xi_2), z_1-z_2\not\in S}d(z_1-z_2)(f(t,\xi_1)+Lg(t,\xi_1)-f(t,\xi_2)-Lg(t,\xi_2))\leq c |\xi_1-\xi_2|
\end{align}
for any $\xi_1$, $\xi_2\in\mathcal{X}$ and $t\geq 0$.  But this is exactly what must be shown to establish the second inequality of \eqref{eq:8}, and thus the claim holds, completing the proof.
\end{proof}

\section{An Example of Traffic Flow Estimation}
\label{sec:an-example-traffic}

Consider a linear freeway with $n$ segments such that traffic flows from segment $i$ to $i+1$, $x_i\in[0,\bar{x}_i]$ is the density of vehicles occupying link $i$, and $\bar{x}_i$ is the capacity of link $i$. An example network is shown in Figure \ref{fig:traffic}. The state-space is then $\mathcal{X}=\prod_{i=1}^n[0,\bar{x}_i]$.

 Associated with each link is a continuously differentiable \emph{demand} function $D_i:[0,\bar{x}_i]\to\mathbb{R}_{\geq 0}$ that is strictly increasing and satisfies $D_i(0)=0$, and a continuously differentiable \emph{supply} function $S_i:[0,\bar{x}_i]\to\mathbb{R}_{\geq 0}$ that is strictly decreasing and satisfies $S_i(\bar{x}_i)=0$. We assume $D'_i(x_i)\geq \nu$  and $S'(x_i)\leq -\omega$ for some $\nu>0$ and some $\omega>0$ for all $i$ and all $x_i\in[0,\bar{x}_i]$ where $'$ denotes differentiation for functions of a scalar argument.

A fraction $\beta_i\in[0,1)$ of the demand on link $i$ is assumed to exit the network via, \emph{e.g.}, unmodeled offramps\footnote{By assuming that the offramp flow is not restricted by downstream capacity, we adopt a \emph{non-first-in-first-out} junction model appropriate if, \emph{e.g.}, dedicated exit lanes exist for the offramp traffic. See \cite{Coogan:2016rp, Kurzhanskiy:2015fj, Lovisari:2014qv} for a discussion of FIFO and non-FIFO assumptions in traffic networks.}. The remaining $1-\beta_i$ fraction of demand is destined for link $i+1$.

We assume $\delta_1(t)$ vehicles per unit of time are available to flow into link $1$ at time $t$ where $\delta(t)\geq 0$ for all $t$. Flow from segment to segment is restricted by upstream demand and downstream supply, and the change in density of a link is governed by mass conservation:
\begin{align}
  \label{eq:38}
\dot{x}_1&= \min\{\delta_1(t),S_1(x_1)\}-p_{1}(x_{1},x_{2})-\beta_1D_1(x_1)\\
  \dot{x}_i&= p_{i-1}(x_{i-1},x_i)-p_{i}(x_{i},x_{i+1})-\beta_iD_i(x_i), \quad i=2,\ldots,n-1\\
  \label{eq:38-3}\dot{x}_n&=p_{n-1}(x_{n-1},x_n)- D_n(x_n)
\end{align}
where, for $i=1,\ldots,n-1$,
\begin{align}
  \label{eq:39}
p_{i}(x_{i},x_{i+1})=\min\{(1-\beta_i)D_i(x_i),S_{i+1}(x_{i+1})\}.
\end{align}
Compactly, we have $\dot{x}=f(t,x)$ where $f$ is defined in accordance with \eqref{eq:38}--\eqref{eq:39}.  It is straightforward to establish positive invariance of the hyper-rectangle $\mathcal{X}$ by verifying that $x_i=0$ implies $\dot{x}_i\geq 0$ and also $x_i=\bar{x}_i$ implies $\dot{x}_i\leq 0$ for all $i=1,\ldots, n$.

Our objective is to determine an observation function $g(x)$ to ensure global exponential convergence of the zero estimation error set for the  the corresponding observer-system interconnection.

\begin{rem}
 Due to the $\min\{\cdot\}$ in  \eqref{eq:38} and \eqref{eq:39}, $f(t,x)$ is not differentiable in $x$. However, $f(t,x)$ is continuous and is \emph{piecewise} differentiable in $x$ \cite{Scholtes:2012fk}. In particular, we interpret $f(t,x)$ as selecting at each $x$ from a finite number of continuously differentiable vector fields (this finite selection at each $x$ is governed by the active minimizers in \eqref{eq:38} and \eqref{eq:39}).  The main results above remain valid in this case by considering the Jacobian matrix induced by each vector field in this finite collection.
\end{rem}

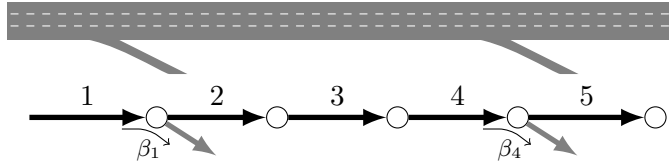
\begin{figure}
  \centering
  \begin{tikzpicture}[scale=.08,yscale=-1,xscale=-1]
\fill[gray] (0,30) rectangle (110,36.2);
\draw[dashed, white] (0,32) -- (110,32);
\draw[dashed, white] (0,34) -- (110,34);
\fill[gray] (27,36) -- (15,42) -- (18,42) .. controls +(1,0) and (28,36) .. (34,36);
\fill[gray] (92,36) -- (80,42) -- (83,42) .. controls +(1,0) and (93,36) .. (99,36);
  \end{tikzpicture}

\tikzstyle{link}=[line width=2pt, ->,>=latex]
\tikzstyle{junc}=[draw,circle,inner sep=1pt,minimum width=8pt]
\tikzstyle{onramp}=[line width=2pt, dashed,->,>=latex]
      \begin{tikzpicture}[xscale=.9]
\node (j0) at (0,0) {};
\node[junc] (j1) at ($(j0)+(0:.8in)$) {};
\node[junc] (j2) at ($(j1)+(0:.7in)$) {};
\node[junc] (j3) at ($(j2)+(0:.7in)$) {};
\node[junc] (j4) at ($(j3)+(0:.7in)$) {};
\node[junc] (j5) at ($(j4)+(0:.8in)$) {};
\draw[link] (j0) to node[above]{$1$} (j1);
\draw[link] (j1) to node[above]{$2$} (j2);
\draw[link] (j2) to node[above]{$3$} (j3);
\draw[link] (j3) to node[above]{$4$} (j4);
\draw[link] (j4) to node[above]{$5$} (j5);
\draw[link,gray] (j1) to ($(j1)+(330:.4in)$);
\draw[link,gray] (j4) to ($(j4)+(330:.4in)$);
\draw[rounded corners, ->] ($(j1.225)+(-.4,-.05)$) -- ($(j1.225) +(0,-.05)$)  node[below,pos=1] {\footnotesize $\beta_1$}-- ($(j1.225)+(320:.4)$) ;
\draw[rounded corners, ->] ($(j4.225)+(-.4,-.05)$) -- ($(j4.225) +(0,-.05)$)  node[below,pos=1] {\footnotesize $\beta_4$}-- ($(j4.225)+(320:.4)$) ;
\end{tikzpicture}
  \caption{\footnotesize A model for a linear freeway consisting of interconnected compartments. The flow of vehicles from compartment to compartment depends on the traffic density of the upstream and downstream links. At the end of segment $i$, a fraction of the vehicles exit the freeway via unmodeled exit ramps as determined by the parameter $\beta_i\geq 0$. If no vehicles exit for some links, $\beta_i=0$, as is the case for $i=2,3,5$ above.}
  \label{fig:traffic}
\end{figure}

Let $\partial_i$ denote differentiation with respect to the $i$th component of $x$. Notice that, for $i=1,\ldots,n-1$,
\begin{align}
  \label{eq:27}
\partial_{i}p_i(x_i,x_{i+1})\geq 0\text{\quad and \quad }\partial_{i+1}p_i(x_i,x_{i+1})\leq 0.   
\end{align}
Define $p_0(t,x_1):=\min\{\delta_1(t),S_1(x_1)\} $. 
The Jacobian of $f(t,x)$, where it exists, is given by 
\begin{align}
\small
  \label{eq:36}
\nonumber  J(t,x):=\frac{\partial f}{\partial x}(t,x)&=
  \underbrace{\begin{pmatrix}
    \partial_1p_0-\partial_1p_1&-\partial_2p_1&0&0&\cdots&0\\
\partial_1 p_1&\partial_2p_1-\partial_2p_2&-\partial_3 p_2&0 &\cdots&0\\
0&\partial_2p_2&\partial_3p_2-\partial_3p_3&-\partial_4p_3&&0\\
\vdots&&&&\ddots&\vdots\\
0&0&\cdots&0&\partial_{n-1}p_{n-1}&\partial_{n}p_{n-1}
  \end{pmatrix}}_{\normalsize =:J_1(t,x)}\\
&+
\underbrace{  \begin{pmatrix}
    -\beta_1D'_1(x_1)&0&0&\cdots&0\\
0&  -\beta_2D'_2(x_2)&&&0\\
0&0&-\beta_3D'_3(x_3)&&0\\
\vdots&&&\ddots&\vdots\\
0&0&0&\cdots&-D'_n(x_n)
  \end{pmatrix}}_{=:J_2(t,x)}.
\end{align}

Let $|\cdot|_1$ denote the vector one-norm, \emph{i.e.}, $|x|_1=\sum_{i=1}^n|x_i|$ for $x\in\mathbb{R}^n$. For $A\in\mathbb{R}^{n\times n}$, the corresponding matrix measure $\mu_1(\cdot)$ is given by \cite{Vidyasagar:2002ly, Desoer:2008bh}
\begin{align}
  \label{eq:28}
  \mu_1(A)=\max_{j=1,\ldots,n}\left([A]_{jj}+\sum_{i\neq j} |[A]_{ij}|\right)
\end{align}
where $[A]_{ij}$ denotes the $ij$-th element of $A$. When all off-diagonal entries of $A$  are nonnegative (such matrices are said to be \emph{Metzler}), \eqref{eq:28} reduces to the maximum column sum of $A$. It is straightforward to confirm using \eqref{eq:27} that $J(t,x)$ is Metzler for all $t\geq 0$ and all $x\in\mathcal{X}$ so that
\begin{align}
  \label{eq:29}
  \mu_1(J(t,x))=\max_{j=1,\ldots,n}\left\{\sum_{i=1}^n[J_1 (t,x)]_{ij}+\sum_{i=1}^n[J_2 (t,x)]_{ij}\right\}.
\end{align}
From \eqref{eq:36}, we see that, for all $t\geq 0$ and all $x\in\mathcal{X}$,
\begin{align}
  \label{eq:30}
  \sum_{i=1}^n[J_1 (t,x)]_{ij}\leq 0\qquad \text{and}\qquad \sum_{i=1}^n[J_2 (t,x)]_{ij}\leq -\beta_j v
\end{align}
for $j=1,\ldots,n$ where we take $\beta_n:=1$. If each $\beta_j>0$, then $  \mu(J(t,x))\leq -\min_{j}\{\beta_j\} v<0$, that is, the dynamics are contractive and the ``observer'' $\dot{\hat{x}}=f(t,\hat{x})$ ensures exponential convergence of the state estimate. 

When $\beta_j=0$ for some $j$, however, we must introduce an observation function to ensure convergence of the observer. Let $J=\{j:\beta_j=0\}$ be the links for which $\beta_j=0$, that is, no traffic exits the network via link $j$.
We propose to construct a differentiable observation function $g_j(x_j)$ for each $j\in J$. For example, $g_j(x_j)$ might be the number of vehicles visible from a roadside video camera for which a portion of link $j$ is visible. If $p_j\in[0,1]$ is the fraction of link $j$ visible in the camera's frame, then a reasonable approximation might be $g_j(x_j)=p_j x_j$, or $g_j$ might be obtained by fitting measured data. We only impose that there exists $\alpha_j>0$ such that $g'(x_j)\geq \alpha_j$ for all $x_j\in [0,\bar{x}_j]$ for all $j$.

We then construct the observation function $g(x)$ elementwise from $g_j$:
\begin{align}
  \label{eq:31}
g&:\mathcal{X}\to \mathbb{R}^n,\\
  [g]_j(x)&=
  \begin{cases}
    g_j(x_j)&\text{if $j\in J$}\\
    0&\text{else}
  \end{cases}
\end{align}
where $[g]_j(x)$ denotes the $j$-th element of $g(x)$. Finally, we choose $L=-I$ where $I$ is the identity matrix. We then have 
\begin{align}
  \label{eq:32}
\sum_{i=1}^n\left[\frac{\partial g}{\partial x}(x)\right]_{ij}=
  \begin{cases}
    g'_j(x_j) &\text{if $j\in J$}\\
    0&\text{else}    
  \end{cases}
\end{align}
for all $x\in\mathcal{X}$. It follows that
\begin{align}
  \label{eq:33}
        \mu_1\left(J(t,x)+L\frac{\partial g}{\partial x}(x)\right)\leq -\max\left\{v\min_{j\not\in J}\beta_j,\min_{j\in J}\alpha_j\right\}<0
\end{align}
for all $t\geq 0$ and all $x\in \mathcal{X}$. Furthermore, we verify that $\mathcal{X}\times\mathcal{X}$ is positively invariant for the observer-system interconnection as follows. Recall that $\mathcal{X}$ is positively invariant for the system, so we need only verify $\hat{x}(t)\in\mathcal{X}$ for all $t$ for all solutions of the state estimate $\hat{x}$. To that end, suppose $\hat{x}_i(t)=0$ for some $i$ at some time $t$. Then, as previously noted, $f_i(t,\hat{x})\geq 0$ where $f_i$ is the $i$-th element of $f$. Further, since $g_i$ is increasing and $x_i\geq \hat{x}_i$, we have $g_i(\hat{x}_i)-g_i(x _i)\leq 0$. Therefore, $\dot{\hat{x}}_i(t)=f_i(t,\hat{x})-(g_i(\hat{x}_i)-g_i(x _i))\geq 0$. An analogous argument implies that if $\hat{x}_i(t)=\bar{x}_i$ for some $t$, then $\dot{\hat{x}}_i(t)\leq 0$ so that $\mathcal{X}\times\mathcal{X}$ is positively invariant for the observer-system interconnection.

The above results can be summarized as follows. For a link $j$ with $\beta_j>0$, some vehicles exit the network upon leaving the link, and this dissipation ensures that the estimated density on this link converges to the actual density. For a link $j$ with $\beta_j=0$, one approach is to measure (a function of) the density on such links. By introducing negative feedback of the difference between this measurement and the estimated measurement, convergence of the observer is guaranteed, and this guarantee is provided via a contraction theoretic argument.

\appendix
\section{Appendix}
\subsection{Clarke Generalized Derivative}
\label{sec:clarke-gener-deriv}
Let $\mathcal{X}\subseteq \mathbb{R}^n$ and $G: \mathcal{X}\to \mathbb{R}$ be a locally Lipschitz continuous function. For any $x\in \mathcal{X}$ and $v\in\mathbb{R}^n$, the \emph{Clarke generalized directional derivative} or \emph{Clarke derivative} of $G$ at $x$ in the direction of $v$ is given by \cite{Clarke:1990kx}
\begin{align}
  \label{eq:17}
  G^\circ(x;v)=\limsup_{z\to x,h\to 0^+}\frac{G(z+hv)-G(z)}{h}.
\end{align}
Let $S\subset \mathcal{X}$ be the measure-zero set on which $G(x)$ is not differentiable. A basic fact of the Clarke derivative is that \cite[Corollary, p. 64]{Clarke:1990kx}
\begin{align}
  \label{eq:18}
  G^\circ(x;v)=\limsup_{z\to x, z\not\in S}\frac{\partial G}{\partial x}(z)v.
\end{align}

\subsection{Contractive Systems}
\label{sec:contractive-systems}
Here, we collect some fundamental results for dynamical systems for which the matrix measure of the Jacobian matrix of the system's vector field is uniformly bounded. In the case when this bound is negative, the system is \emph{contractive}, that is, the distance between any two trajectories is exponentially decreasing.
\begin{prop}
\label{prop:1}
  Consider
  \begin{align}
    \label{eq:4}
    \dot{x}=F(t,x)
  \end{align}
with $t\in[0,\infty)$ and $x\in\mathcal{X}\subset \mathbb{R}^n$ for positively invariant and convex $\mathcal{X}$, and assume $F(t,x)$ is differentiable in $x$, and $F(t,x)$ and $\frac{\partial F}{\partial x}(t,x)$ are continuous in $(t,x)$.   Let $|\cdot|$ be a norm on $\mathbb{R}^n$ and $\mu(\cdot)$ its associated matrix measure. If, for some $c\in\mathbb{R}$,
\begin{align}
  \label{eq:7}
  \mu\left(\frac{\partial F}{\partial x}(t,x)\right)\leq c,\quad \forall x\in \mathcal{X}\quad \forall t\geq 0,
\end{align}
then for any $x_1(t)$ and $x_2(t)$ solutions of \eqref{eq:4} and for all $t\geq t_0 \geq 0$,
\begin{align}
  \label{eq:9}
  |x_1(t)-x_2(t)|\leq e^{c(t-t_0)}|x_1(t_0)-x_2(t_0)|.
\end{align}
\begin{proof}
  A proof is found in \cite [proof of Theorem 1]{Sontag:2010fk}. Note that Theorem 1 of \cite{Sontag:2010fk} assumes $c<0$, but the result and proof technique remains valid for all $c\in\mathbb{R}$.
\end{proof}

\end{prop}

\begin{corollary}
\label{cor:1}
    Under the hypotheses of Proposition \ref{prop:1}, let $S\subset \mathbb{R}^n$ be the measure-zero set where $|\cdot|$ is not differentiable, and let $x_1(t)$ and $x_2(t)$ be solutions of \eqref{eq:4}. If $x_1(t_0)-x_2(t_0)\not \in S$ for some $t_0\geq 0$, then
    \begin{align}
      \label{eq:15}
      \frac{d}{dt}(|x_1(t)-x_2(t)|) \big |_{t=t_0}\leq c |x_1(t_0)-x_2(t_0)|.
    \end{align}
\end{corollary}
\begin{proof}
  Let $a(t)=|x_1(t)-x_2(t)|$ and let $b(t)=e^{c(t-t_0)} |x_1(t_0)-x_2(t_0)|$  so that $a(t_0)=b(t_0)$. Suppose \eqref{eq:15} did not hold, \emph{i.e.}, $\dot{a}(t_0)> \dot{b}(t_0)$. The functions $a(t)$ and $b(t)$ are continuously differentiable for $t$ near $t_0$, and therefore there exists $\tau>0$ such that $\dot{a}(t) >\dot{b}(t)$ for all $t\in[t_0,t_0+\tau]$. By the mean-value theorem, there exists $t'\in[t_0,t_0+\tau]$ such that $\dot{a}(t')-\dot{b}(t')=\frac{1}{\tau}({a(t_0+\tau)-b(t_0+\tau)-a(t_0)+b(t_0)})=\frac{1}{\tau}(a(t_0+\tau)-b(t_0+\tau))$. Since $\dot{a}(t')>\dot{b}(t')$, we then have that $a(t_0+\tau)>b(t_0+\tau)$. But Proposition \ref{prop:1} guarantees $a(t)\leq b(t)$ for all $t\geq t_0$, a contradiction.
\end{proof}

\begin{corollary}
\label{cor:2}
  Under the hypotheses of Proposition \ref{prop:1}, let $S\subset \mathbb{R}^n$ be the measure-zero set where $|\cdot|$ is not differentiable, and for $z\not\in S$, define $d(z):=\frac{\partial}{\partial z}|z|$. For all $x_1,x_2\in\mathcal{X}$ and all $t\geq 0$,
  \begin{align}
    \label{eq:11}
\limsup_{(z_1,z_2)\to (x_1,x_2), z_1-z_2\not\in S}d(z_1-z_2)(F(t,x_1)-F(t,x_2))\leq c |x_1-x_2|.
  \end{align}

\end{corollary}
\begin{proof}
  Suppose not, and let $(z_1^i)_{i=1}^\infty$ and $(z_2^i)_{i=1}^\infty$ be sequences converging to some $x_1$ and $x_2$, respectively, such that $z_1^i-z_2^i\not \in S$ for all $i$ and there exists $\epsilon_1>0$ such that
  \begin{align}
    \label{eq:12}
    \lim_{i\to \infty} d(z_1^i-z_2^i)(F(t_0,x_1)-F(t_0,x_2)) \geq c |x_1-x_2|+\epsilon_1
  \end{align}
for some $t_0\geq 0$. 

Because $F(t_0,\cdot )$ and $|\cdot|$ are continuous functions and $(z_1^i,z_2^i)\to  (x_1,x_2)$, for any $\epsilon_2\in (0,\epsilon_1)$, there exists sufficiently large $i^*$ such that
\begin{align}
  \label{eq:16}
  d(z_1^{i^*}-z_2^{i^*})(F(t_0,z_1^{i^*})-F(t_0,z_2^{i^*})) \geq c |z_1^{i^*}-z_2^{i^*}|+\epsilon_2.
\end{align}

Now let $z_1(t)$ and $z_2(t)$ denote solutions to \eqref{eq:4} with $z_1(t_0)=z_1^{i^*}$, $z_2(t_0)=z_2^{i^*}$. Then the left-hand side of \eqref{eq:16} is equal to  $\frac{d}{dt}(|z_1(t)-z_2(t)|) \big |_{t=t_0}$. However, Corollary \ref{cor:1} guarantees $\frac{d}{dt}(|z_1(t)-z_2(t)|) \big |_{t=0}\leq c |z_1(t_0)-z_2(t_0)|= c|z_1^{i^*}-z_2^{i^*}|$, a contradiction.
\end{proof}

\bibliographystyle{ieeetr}
\bibliography{biblio.bib}

\begin{thebibliography}{10}

\bibitem{LOHMILLER:1998bf}
W.~Lohmiller and J.-J.~E. Slotine, ``On contraction analysis for non-linear
  systems,'' {\em Automatica}, vol.~34, no.~6, pp.~683--696, 1998.

\bibitem{Sontag:2010fk}
E.~D. Sontag, ``Contractive systems with inputs,'' in {\em Perspectives in
  Mathematical System Theory, Control, and Signal Processing}, pp.~217--228,
  Springer, 2010.

\bibitem{Forni:2012qe}
F.~Forni and R.~Sepulchre, ``A differential {L}yapunov framework for
  contraction analysis,'' {\em IEEE Transactions on Automatic Control},
  vol.~59, pp.~614--628, March 2014.

\bibitem{Aghannan:2003wj}
N.~Aghannan and P.~Rouchon, ``An intrinsic observer for a class of {L}agrangian
  systems,'' {\em IEEE Transactions on Automatic Control}, vol.~48, no.~6,
  pp.~936--945, 2003.

\bibitem{Bonnabel:2010fx}
S.~Bonnabel, ``A simple intrinsic reduced-observer for geodesic flow $ $,''
  {\em IEEE Transactions on Automatic Control}, vol.~55, no.~9, pp.~2186--2191,
  2010.

\bibitem{Sanfelice:2012uf}
R.~G. Sanfelice and L.~Praly, ``Convergence of nonlinear observers on
  $\mathbb{R}^n$ with a {R}iemannian metric (part i),'' {\em IEEE Transactions
  on Automatic Control}, vol.~57, no.~7, pp.~1709--1722, 2012.

\bibitem{Dani:2015ng}
A.~P. Dani, S.-J. Chung, and S.~Hutchinson, ``Observer design for stochastic
  nonlinear systems via contraction-based incremental stability,'' {\em IEEE
  Transactions on Automatic Control}, vol.~60, no.~3, pp.~700--714, 2015.

\bibitem{Vidyasagar:2002ly}
M.~Vidyasagar, {\em Nonlinear System Analysis}.
\newblock Society for Industrial and Applied Mathematics, 2002.

\bibitem{Desoer:2008bh}
C.~Desoer and M.~Vidyasagar, {\em Feedback systems: Input-output properties}.
\newblock Society for Industrial and Applied Mathematics, 2008.

\bibitem{Clarke:1990kx}
F.~H. Clarke, {\em Optimization and nonsmooth analysis}, vol.~5.
\newblock {SIAM}, 1990.

\bibitem{khalil}
H.~K. Khalil, {\em Nonlinear Systems}.
\newblock Prentice Hall, third~ed., 2002.

\bibitem{Coogan:2016rp}
S.~Coogan, M.~Arcak, and A.~A. Kurzhanskiy, ``Mixed monotonicity of partial
  first-in-first-out traffic flow models,'' in {\em IEEE Conference on Decision
  and Control}, pp.~7611--7616, 2016.

\bibitem{Kurzhanskiy:2015fj}
A.~A. Kurzhanskiy and P.~Varaiya, ``Traffic management: An outlook,'' {\em
  Economics of Transportation}, vol.~4, no.~3, pp.~135--146, 2015.

\bibitem{Lovisari:2014qv}
E.~Lovisari, G.~Como, A.~Rantzer, and K.~Savla, ``Stability analysis and
  control synthesis for dynamical transportation networks,'' {\em
  arXiv:1410.5956}, 2014.

\bibitem{Scholtes:2012fk}
S.~Scholtes, {\em Introduction to piecewise differentiable equations}.
\newblock Springer, 2012.

\end{thebibliography}

\end{document}